\documentclass[envcountsame]{llncs}

\usepackage{enumerate}
\usepackage{tikz}
\usetikzlibrary{snakes}

\usepackage{amsmath,amstext,amssymb}
\usepackage{xspace}   

\spnewtheorem{reduction}{Reduction}{\bfseries}{}

\begin{document}

\title{Towards optimal kernel for connected vertex cover in planar graphs\thanks{L. Kowalik and M. Pilipczuk are supported by the National Science Centre (grants N206 567140 and N206 355636). Additionally M. Pilipczuk is supported by the Foundation for Polish Science. K. Suchan is supported by Basal-CMM, Anillo ACT88 and Fondecyt 11090390 programs of Conicyt, Chile.}}
\titlerunning{Towards optimal kernel for connected vertex cover in planar graphs}

\author{
  \L{}ukasz Kowalik\inst{1}
  \and
  Marcin Pilipczuk\inst{1}
  \and
  Karol Suchan\inst{2,}\inst{3}
}
\authorrunning{Kowalik et al.}

\institute{Institute of Informatics, University of Warsaw, Poland\\
  \email{\{kowalik@,malcin@\}mimuw.edu.pl}
  \and
  Universidad Adolfo Ib\'{a}\~{n}ez, Chile\\
  \email{karol.suchan@uai.cl}
  \and
  WMS, AGH - University of Science and Technology, Poland.
}

\maketitle

\begin{abstract}
  We study the parameterized complexity of the connected version of the vertex cover problem, where the solution set has to induce a connected subgraph.
  Although this problem does not admit a polynomial kernel for general graphs (unless $\rm{NP}\subseteq \rm{coNP}/poly$), for planar graphs Guo and Niedermeier [ICALP'08] showed a kernel with at most $14k$ vertices, subsequently improved by Wang et al. [MFCS'11] to $4k$. 
  The constant $4$ here is so small that a natural question arises: could it be already an optimal value for this problem? In this paper we answer this quesion in negative: we show a $\frac{11}{3}k$-vertex kernel for  {\sc Connected Vertex Cover} in planar graphs. We believe that this result will motivate further study in search for an optimal kernel.
\end{abstract}

\section{Introduction}

Many NP-complete problems, while most likely cannot be fully solved efficiently, admit kernelization algorithms, i.e.\ efficient algorithms which replace input instances with an equivalent, but often much smaller ones. More precisely, a {\em kernelization algorithm} takes an instance $I$ of size $n$ and a parameter $k\in\mathbb{N}$, and after a time polynomial in $n$ it outputs an instance $I'$ (called a {\em kernel}) with a parameter $k'$ such that $I$ is a yes-instance iff $I'$ is a yes instance, $k'\le k$, and $|I'|\le f(k)$ for some function $f$ depending only on $k$. The most desired case is when the function $f$ is polynomial, or even linear (then we say that the problem admits a polynomial or linear kernel). In such a case, when the parameter $k$ is relatively small, the input instance, possibly very large, is ``reduced'' to a small one. Intuitively, kernelization aims at finding the core difficulty in the input instance. The output instance can be then processed in many ways, including approximation algorithms or heuristics. For small values of $k$ and small kernels, one can often even use an exact (exponential-time) algorithm.

A typical example of the above phenomenon is the well-known {\sc Vertex Cover} problem, which admits a kernel with at most $2k$ vertices~\cite{ckj:vc}, where $k$ is the size of an optimum vertex cover in the input instance. However, for some problems, reducing to a linear number of vertices seems a hard task, e.g.\ for the {\sc Feedback Vertex Set} problem, the best known result is the $4k^2$-vertex kernel of Thomass\'{e}~\cite{fvs:quadratic-kernel}. Even worse, there are many natural problems (examples include {\sc Dominating Set} or {\sc Steiner Tree}) for which it is proved that they do not admit a polynomial kernel, unless some widely believed complexity hypothesis fails ($FPT \neq W[2]$ in the first case and $\rm{NP}\varsubsetneq \rm{coNP}/poly$ in the second one). This motivates investigation of polynomial kernels in natural restrictions of graph classes. Note that it is of particular interest to guarantee the output instance belongs to the same class as the input instance.

A classic example is the $335k$-vertex kernel for the {\sc Dominating Set} problem in planar graphs due to Alber et al.~\cite{afn:planar-domset} (note that in planar graphs the number of edges is linear in the number of vertices, thus kernels with linear number of vertices are in fact linear kernels). Later, this work was substantially generalized (see e.g.~\cite{fomin:bidim-kernels}), and researchers obtained a number of linear kernels for planar graphs. Let us mention the $112k$-vertex kernel for  {\sc Feedback Vertex Set} by Bodlaender and Penninkx~\cite{bp:planar-fvs}, the linear kernel for {\sc Induced Matching} by Moser and Sikdar~\cite{ms:indmatch} or the $624k$-vertex kernel for the {\sc Maximum Triangle Packing} by Guo and Niedermeier~\cite{guon:planarkernels}.

Observe that the constants in the linear functions above are crucial: since we deal with NP-complete problems, in order to find an exact solution in the reduced instance, most likely we need exponential time (or at least superpolynomial, because for planar graphs $2^{O(\sqrt{k})}$-time algorithms are often possible), and these constants appear in the exponents.
Motivated by this, researchers seek for linear kernels with constants as small as possible. For example, by now there is known a $67k$-vertex kernel for {\sc Dominating Set}~\cite{cfkx:duality-and-vertex}, a $28k$-vertex kernel for  {\sc Induced Matching}~\cite{ekkw:indmatch} or a $75k$-vertex kernel for {\sc Maximum Triangle Packing}~\cite{mfcs}.

In this work, we study the {\sc Connected Vertex Cover} problem, a variant of the classical {\sc Vertex Cover}: we are given a planar graph $G=(V,E)$ and a parameter $k$, and we ask whether there exists a vertex cover $S$ (i.e.\ a set $S\subseteq V$ such that every edge of $G$ has an endpoint in $S$) of size $k$ which induces a connected subgraph of $G$. This problem is NP-complete also in planar graphs \cite{garey:planarcvc}, and, contrary to its simpler relative, it does not admit a polynomial kernel in arbitrary graphs~\cite{dom:colorsandids}. However, Guo and Niedermeier~\cite{guon:planarkernels} showed a $14k$-vertex kernel for planar graphs. Very recently, it was improved to $4k$ by Wang et al.~\cite{mfcs}. The constant $4$ here is already so small that a natural question arises: could it be the optimal value for this problem? In this paper we answer this question in negative: we show a $\frac{11}{3}k$-vertex kernel for  {\sc Connected Vertex Cover} in planar graphs.

Let us recall that in the analysis of the $4k$-vertex kernel by Wang et al. \cite{mfcs},
the vertices of the graph are partitioned into three parts: vertices of degree one, the solution, and the rest of the graph, and it is proven that, after applying a few reduction rules,
the sizes of these parts can be bounded by $k$, $k$ and $2k$, respectively.
We present (in Lemma \ref{lem:key}) a deeper analysis of these bounds and we show that
an instance where all the bounds are close to being tight is somewhat special.
This analysis is the main technical contribution of this paper.
We believe that this result will motivate further study in search for an optimal kernel.

\paragraph{Organization of the paper}
The paper is organized as follows. In Section~\ref{sec:alg} we present our kernelization algorithm along with a proof of its correctness. 
In Section~\ref{sec:analysis} we show that our algorithm outputs a kernel with the number of vertices bounded by $\frac{11}{3}k$.
Finally, in Section~\ref{sec:example} we describe an example which shows that our analysis is tight (and hence improving on the kernel size would require adding new reduction rules to the algorithm).

\paragraph{Terminology and notation}
We use standard fixed parameter complexity and graph theory terminology, see e.g.~\cite{downey-fellows:book,diestel}.
For brevity, we call a vertex of degree $d$ a {\em $d$-vertex} and if a vertex $v$ has a $d$-vertex $w$ as a neighbor we call $w$ a $d$-neighbor.
By $N_G(v)$ we denote the set of neighbors of $v$ and we omit the subscript when it is clear from the context.
For a graph $G$ and a subset of its vertices $S$, by $G[S]$ we denote the subgraph of $G$ induced by $S$.


\section{Algorithm}
\label{sec:alg}

In this section we present our kernelization algorithm. It works in three phases. In what follows, $(G_0,k_0)$ denotes the input instance.

\subsection{Phase 1}

Phase 1 is a typical kernelization algorithm. A set of rules is specified and in each rule the algorithm searches the graph for a certain configuration. If the configuration is found, the algorithm performs a modification of the graph (sometimes, also of the parameter $k$), typically decreasing the size of the graph. Each rule has to be {\em correct}, which means that the new graph is planar and the graph before application of the rule has a connected vertex cover of size $k$ if and only if the new graph has a connected vertex cover of size equal to the new value of $k$. We apply the rules in order, i.e.\ Rule $i$ can be applied only if, for every $j<i$, Rule $j$ does not apply.
The first three rules come from~\cite{mfcs}. 

\newcounter{savedenumi}
\begin{enumerate}[Rule 1. ]
\item \label{r:1-vtx} If a vertex $v$ has more than one 1-neighbors, then remove all these neighbors
except for one.
\item For a 2-vertex $v$ with $N(v) = \{u,w\}$ and $uw \in E(G)$, contract the edge $uw$ and decrease the parameter $k$ by one.\label{r:2->1edge}
\item For a 2-vertex $v$ with $N(v) = \{u,w\}$ and $uw \not\in E(G)$, if $v$ is not a cut-vertex, then
remove $v$ and add a new 1-neighbor to each of $u$ and $w$; otherwise, contract the edge $uv$
and decrease the parameter $k$ by one.
\label{r:2->1}
\item If there is an edge $uv$ such that both $u$ and $v$ have a 1-neighbor, then remove the 1-neighbor of $u$, contract $uv$ and decrease the parameter $k$ by one.\label{r:1-1}
\item \label{r:3-vtx} If there is a 3-vertex $v$ with $N(v)=\{x,y,z\}$ and such that $\deg(z)=1$, then remove $v$ and $z$, add an edge $xy$ if it was not present before, and decrease
the parameter $k$ by one. 
\item \label{r:K_2,3} If there are two 3-vertices $a$ and $b$ with a common neighborhood $N(a)=N(b)=\{x,v,y\}$, and such that removing any two vertices from $\{x,v,y\}$ makes the graph disconnected, then remove $a$ and add three $1$-vertices $x'$, $v'$ and $y'$, adjacent to $x$, $v$ and $y$, respectively.
\item \label{r:two-triangles} If there is a 3-vertex $a$ with $N(a)=\{x,v,y\}$ and such that $v$ is a 4-vertex with $N(v)=\{x,a,y,q\}$, where $q$ is a 1-vertex, then remove vertices
$a$, $v$ and $q$, add an edge $xy$ as well as two $1$-vertices $x'$ and $y'$, connected to $x$ and $y$ respectively, and decrease the parameter $k$ by one.
\label{r:last}
\setcounter{savedenumi}{\value{enumi}}
\end{enumerate}

\begin{figure}
\begin{center}
\includegraphics{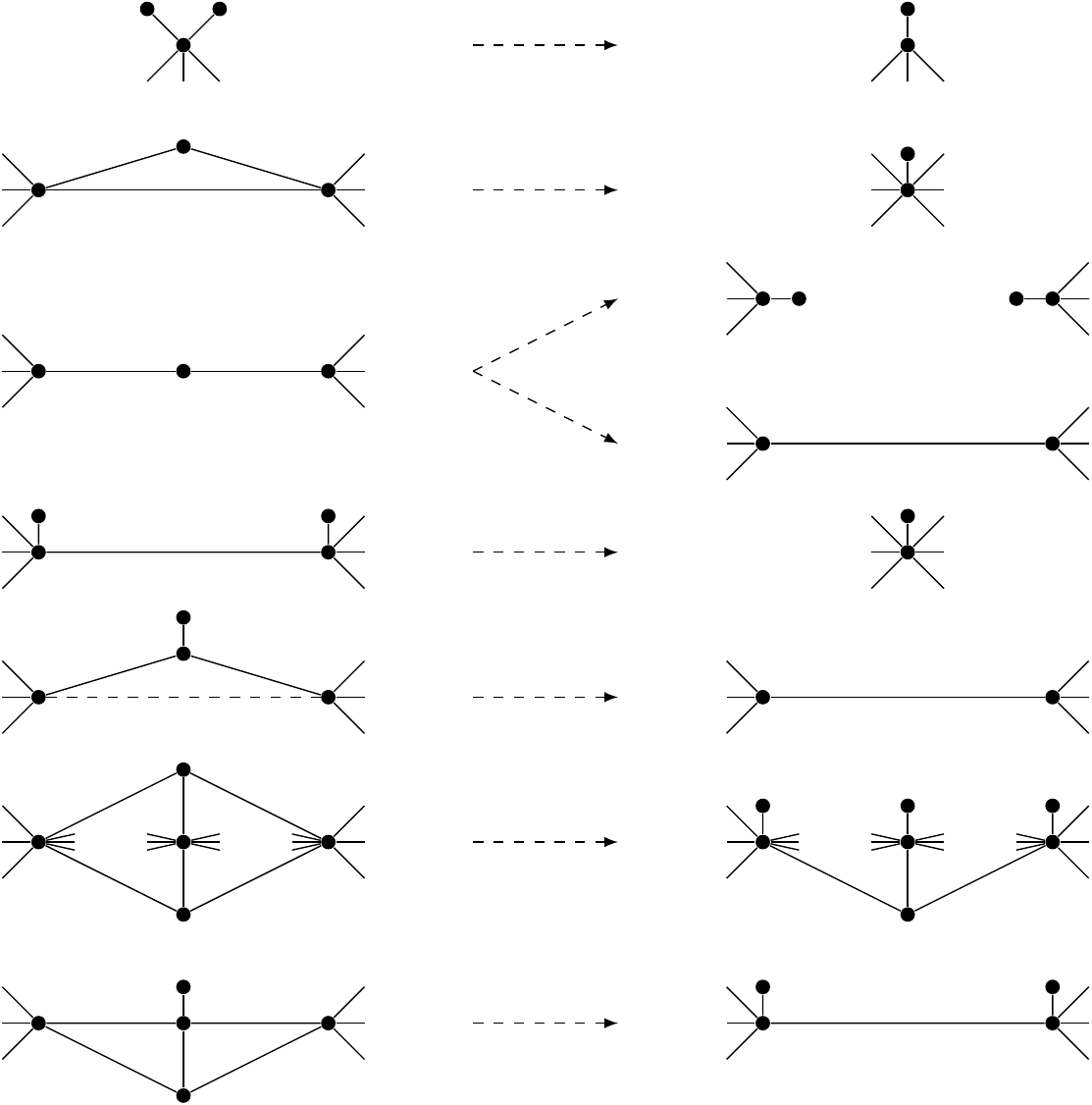}
\caption{Rules used in Phase 1.}
  \label{fig:rules}
  \end{center}
  \end{figure}

\begin{lemma}
 Rules 1-\ref{r:last} are correct. 
\end{lemma}

\begin{proof}
The correctness of Rules \ref{r:1-vtx}-\ref{r:2->1} were proven in \cite{mfcs}.
Rule \ref{r:1-1} is obvious, as both $u$ and $v$ need to be included in any connected vertex cover of the graph.
In the proofs for Rules \ref{r:3-vtx}-\ref{r:last}, by $G$ and $G'$ we denote the graph respectively before and after the currently considered rule was applied.

\noindent\underline{Rule \ref{r:3-vtx}.} Let $S'$ be a connected vertex cover in $G'$. We claim that $S = S' \cup \{v\}$ is a connected vertex cover in $G$.
As $v \in S$, $S$ is a vertex cover of $G$. Moreover, as $xy \in E(G')$, either $x$ or $y$ belongs to $S'$, and if both $x$ and $y$ belong to $S'$,
then they remain connected in $G[S]$ via the vertex $v$. We infer that $G[S]$ is connected.

In the other direction, let $S$ be a minimum connected vertex cover of $G$. Clearly $v \in S$ and $z \notin S$.
As Rules \ref{r:1-vtx}-\ref{r:2->1} are not applicable, the vertices $x$ and $y$ have degrees at least $3$. Therefore, as $G[S]$ is connected, $x$ or $y$ belongs to $S$;
without loss of generality assume that $x \in S$. It follows that $S' = S \setminus \{v\}$ is a vertex cover of $G'$. Moreover, if $y \in S$, then $x$ and $y$
remain connected in $G[S']$ via the edge $xy$. Thus $G'[S']$ is connected and the proof of the correctness of Rule \ref{r:3-vtx} is finished.

\noindent\underline{Rule \ref{r:K_2,3}.} We first note that any connected vertex cover in $G'$ is also a connected vertex cover in $G$, since it needs to contain $x$, $v$, $y$
due to new vertices $x'$, $v'$, $y'$.

In the other direction, let $S$ be a minimum connected vertex cover in $G$. Note that $S$ needs to contain at least two vertices out of $\{x,v,y\}$, since removing $\{x,v,y\} \setminus S$ cannot
disconnect $G$. 
If $a \in S$ while $b \notin S$, then $(S \setminus \{a\}) \cup \{b\}$ is also a minimum connected vertex cover of $G$. Thus we can assume that if $a \in S$ then $b \in S$ as well.
We infer that if $\{x,v,y\} \subseteq S$, then $S \setminus \{a\}$ is a connected vertex cover in $G'$ too.
Otherwise, without loss of generality let $v \notin S$. As $S$ is a vertex cover, $a,b \in S$.
As $N_G(a) = N_G(b)$, we have that $G[S \setminus \{a\}]$ is connected. Thus $(S \setminus \{a\}) \cup \{v\}$ is a connected vertex cover of both $G$ and $G'$, and Rule \ref{r:K_2,3} is correct.

\noindent\underline{Rule \ref{r:two-triangles}.} Let $S'$ be a minimum connected vertex cover in $G'$. The vertices $x'$ and $y'$ ensure that $x,y \in S'$ and $x',y' \notin S'$.
Then clearly $S = S' \cup \{v\}$ is a connected vertex cover in $G$, as $x$ and $y$ remain connected in $G[S]$ via the vertex $v$.

In the other direction, let $S$ be a minimum connected vertex cover in $G$. Clearly $v \in S$ and $q \notin S$. Since Rule \ref{r:2->1} is not applicable,
the degrees of $x$ and $y$ are at least $3$. As $G[S]$ is connected, we infer that $x$ or $y$ belongs to $S$; without loss of generality we can assume that $x \in S$.
We claim that $S' = (S \setminus \{v,a\}) \cup \{y\}$ is a connected vertex cover of $G'$ of size at most $|S|-1$. If $y \in S$, the statement is obvious.
Otherwise, since $S$ is a vertex cover of $G$, we have that $a \in S$ and $|S'| = |S|-1$. This finishes the proof of the lemma.
\end{proof}

\subsection{Phase 2}

In what follows, the graph obtained after Phase 1 is denoted by $G_1$. Graph $G_1$, similarly as the reduced graph in~\cite{mfcs}, does not contain 2-vertices and every vertex has at most one 1-neighbor. The goal of Phase 2 is to decrease the number of vertices in the graph by replacing some pairs of 1-vertices by 2-vertices using the following rule (a kind of inverse of~\ref{r:2->1}):

\begin{enumerate}[Rule 1.]

\setcounter{enumi}{\value{savedenumi}}

\item If there are two vertices $u$ and $v$, both with 1-neigbors, say $x_u$ and $x_v$, and $u$ and $v$ are incident with the same face then identify $x_u$ and $x_v$.
\label{r:1->2}
\end{enumerate}

In the above rule we assume that we have a fixed plane embedding of $G_1$ (if the input graph is not given as a plane embedding, it can be found in linear time from the set of edges by an algorithm of Hopcroft and Tarjan~\cite{tarjan}). Note that Rule~\ref{r:1->2} preserves planarity. Observe also that since the graph before application of this rule is connected, then the in new graph the vertex which appears after identifying $x_u$ and $x_v$ is not a cut-vertex. Moreover, Rule \ref{r:1-1} guarantees that $uv \notin E(G_1)$.
It follows that the correctness of Rule~\ref{r:2->1} implies the corectness of Rule~\ref{r:1->2}.

In order to get our bound on the kernel we do not apply Rule~\ref{r:1->2} greedily, but we maximize the number of times it is applied. 
To this end, an auxiliary graph $G_M$ is built. Let $S_1$ be the set of vertices of graph $G_1$ that have a 1-neighbor. The vertex set of $G_M$ is equal to $S_1$. Two vertices $u$ and $v$ of $S_1$ are adjacent in $G_M$ if and only if $u$ and $v$ are incident to the same face in $G_1$. (Note that $G_M$ does not need to be planar.) Our algorithm finds a maximum matching $M^*_0$ in $G_M$ in polynomial time. Next we modify $M^*_0$ to get another matching $M^*$ of the same size. We start with $M^*=\emptyset$. Then, for each face $f$ of $G_1$ we find the set $M_f$ of all edges $uv$ of $M^*_0$ such that both $u$ and $v$ are incident to $f$. Let $v_1,\ldots,v_{2|M_f|}$ be the vertices of $V(M_f)$ in the clockwise order around $f$. We add the set $\{v_1v_2, \ldots, v_{2|M_f|-1}v_{2|M_f|}\}$ to $M^*$. It is clear that after applying this procedure to all the faces of $G_1$ we have $|M^*|=|M_0^*|$. Moreover, the graph $G_1 \cup M^*$ is planar, and we can extend the plane embedding of $G_1$ to a plane embedding of  $G_1 \cup M^*$. It follows that Rule~\ref{r:1->2} can be applied $|M^*|$ times according to the matching $M^*$. The time needed to perform Phase 2 is dominated by finding the matching $M^*_0$ which can be done in $O(\sqrt{|V(G_M)|}\cdot |E(G_M)|) = O(n^{2.5})$ using the Micali-Vazirani algorithm~\cite{micali-vazirani}.

\subsection{Phase 3}

Let $(G_2,k_2)$ be the instance obtained after Phase 2. 
In the next section we show that if $G_2$ contains a connected vertex cover of size $k_2$, then $|V(G_2)| \le \frac{11}{3} k_2$.
Together with the correctness of Rules 1-\ref{r:1->2} this justifies the correctness of the last step of our kernelization algorithm: if $|V(G_2)| > \frac{11}{3} k_2$ the algorithm reports that in the input graph $G_0$ does not contain a connected vertex cover of size $k_0$.

\section{Analysis}
\label{sec:analysis}

Let $S$ be a minimum connected vertex cover in $G_1$. Clearly, $S$ is a connected vertex cover in $G_2$, as $S$ does not contain any $1$-vertices. Moreover,
the correctness of Rule \ref{r:2->1} ensures that $S$ is also a minimum one.
The goal of this section is to show that $|V(G_2)| \le \frac{11}{3} |S|$.
However, most of the time we fill focus on the graph $G_1$.

Observe that we can assume that every vertex in $S$ has degree at least 3: this is not the case only in the trivial case when $G_1$ is a single edge.
Note that $V(G_1)\setminus S$ is an independent set. We denote it by $I$. 
The set $I$ is further partitioned into three subsets: $I_1$, $I_3$ and $I_{\ge 4}$ which contain vertices of $I$ of degree 1, 3 and at least 4 respectively.
Note that each neighbor of a vertex in $I_1$ belongs to $S$. We denote by $S_1$ the set vertices in $S$ which have a neigbor in $I_1$ and let $S_{\ge 3} = S \setminus S_1$. 

Roughly, we want to show that $S$ is a big part of $V(G_1)$. Following~\cite{mfcs}, we can bound $|I|$ as follows. Consider the bipartite planar graph $B$ which consists of the edges of $G_1$ between $S$ and $I\setminus I_1$. Then $3|I\setminus I_1| \le |E(B)| < 2 (|S|+|I\setminus I_1|)$, where the second inequality follows from the well-known fact that in a simple bipartite planar graph the number of edges is smaller than twice the number of vertices. This implies that $|I\setminus I_1| < 2|S|$. Since in $G_1$ every vertex has at most one 1-neighbor, so $|I_1|=|S_1|\le |S|$. Hence $|I|<3|S|$. However, with our additional rules, this inequality is not tight. There are three events which make $|I|$ even smaller than $3|S|$. Obviously, this happens when the matching $M^*$ is large. Second good event is when $S_{\ge 3}$ is large, because this means that $|S_1|$ is much smaller than $|S|$, so we can improve our bound on $|I_1|$. Finally, it is also good when the set $I_{\ge 4}$ is large, because then we get a better bound on $|I\setminus I_1|$. We will show that at least one of  these three situations happen in $G_1$. This is guarnteed by the following lemma.

\begin{lemma}
\label{lem:key}
 $|S_{\ge 3}| + |I_{\ge 4}| + |M^*| \ge |S|/3$.
\end{lemma}

In order to prove Lemma~\ref{lem:key} we need the following auxiliary result.

\begin{lemma}
 \label{lem:matching}
 In any simple planar graph which contains $n_{\ge 3}$ vertices of degree at least 3 there is a matching of cardinality at least $n_{\ge 3} / 3$.
\end{lemma}

\begin{proof}
Let $G$ be an arbitrary planar graph. Let $n_{\le 2}$ denote the number of vertices of $G$ of degree at most 2 and $V_{\ge 3}(G)$ --- the set of vertices of degree at least $3$ in $G$.
We use induction on $n_{\le 2}$.
If $n_{\le 2}=0$, the lemma follows by a result of Nishizeki and Baybars~\cite{nishizeki} who proved that any $n$-vertex planar graph of minimum degree 3 contains a matching of size at least $\tfrac{1}{3}(n+2)$.
Now assume $n_{\le 2}>0$. Let $v$ be an arbitrary vertex of degree at most 2. There are three cases to consider.

Case 1: $\deg(v)=0$. Then we remove $v$ and we apply the induction hypothesis. 

Case 2: $\deg(v)=1$. Let $w$ be its only neighbor. If $\deg(w)\ne 3$ we can just remove $v$ and use the induction hypothesis. If $\deg(w)=3$ then let $G'$ be the graph obtained from $G$ by removing $v$ and $w$. $G'$ has at least $n_{\ge 3} - 3$ vertices of degree at least 3 ($V_{\ge 3}(G) \setminus V_{\ge 3}(G')$ may contain only $w$ and its two neighbours different than $v$), so by the induction hypothesis $G'$ has a matching $M_0$ of size at least $n_{\ge 3}/3 - 1$. Then $M_0\cup\{vw\}$ is the desired matching in $G$.

Case 3: $\deg(v)=2$. Let $N(v)=\{x,y\}$. If $xy \not \in E(G)$ then we obtain $G'$ from $G$ by removing $v$ and adding an edge $xy$.
Note that $G'$ is simple, planar and $V_{\ge 3}(G) = V_{\ge 3}(G')$.
By the induction hypothesis $G'$ has a matching $M_0$ of size at least $n_{\ge 3} / 3$. If $xy\not\in M_0$ then $M_0$ is the desired matching in $G$, otherwise we just use $M_0 \setminus \{xy\} \cup \{xv\}$.
Hence we are left with the case when $xy \in E(G)$. If $\deg(x) \neq 3$ and $\deg(y) \neq 3$ then we can just remove $v$ and use the matching from the induction hypothesis.
Hence, w.l.o.g.\ $\deg(x)=3$. Then let $G'$ be the graph obtained from $G$ by removing $v$ and $x$. $G'$ has at least $n_{\ge 3} - 3$ vertices of degree at least 3
($V_{\ge 3}(G) \setminus V_{\ge 3}(G')$ may contain only $x$, $y$ and the third neighbour of $x$ different than $y$ and $v$),
so by the induction hypothesis $G'$ has a matching $M_0$ of size at least $n_{\ge 3}/3 - 1$. Then $M_0\cup\{vx\}$ is the desired matching in $G$.
\qed
\end{proof}

\begin{proof}[of Lemma~\ref{lem:key}]
Let us consider an auxiliary graph $W$. Its vertex set consists of three types of vertices: 
elements of $S$, $I_{\ge 4}$ and additionally, for each $v\in S_{\ge 3}$ we add three vertices 
$v_1$, $v_2$ and $v_3$. The edge set can be contructed as follows. 
First, for every $v,w \in S \cup I_{\ge 4}$, we add $vw$ to $W$ whenever $vw\in E(G_1)$. Second, for each $v\in S_{\ge 3}$ we add three edges $vv_1$, $vv_2$ and $vv_3$. Finally, we consider faces of $G_1$, one by one. For each such face $f$ we do the following. Let $u_1, \ldots, u_{\ell}$ be the vertices of $V(W)$ incident to $f$, in clockwise order. Then, if $\ell>1$, we add edges $u_1u_2, u_2u_3, \ldots u_{\ell-1}u_{\ell},u_{\ell}u_1$. In this process we do not create double edges, i.e.\ if an edge is already present in the graph, we do not add another copy of it. It is clear that $W$ is a planar graph. Let us prove the following claim.

\begin{figure}
 \begin{center}
\includegraphics{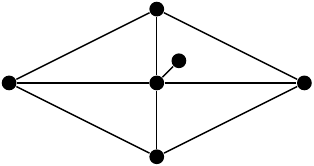}
\caption{A 6-vertex graph with minimum connected vertex cover of size 3.}
\label{fig:exception}
\end{center}
\end{figure}

{\bf Claim:} Every vertex $v\in S_1$ has degree at least 3 in graph $W$, unless $G_1$ is the graph from Fig.~\ref{fig:exception}.

By Rule~\ref{r:1-vtx} there is exactly one 1-neighbor of $v$ in $G_1$, let us call it $q$.
By Rule~\ref{r:3-vtx}, $\deg_{G_1}(v) \ge 4$, so we have $|N_{G_1}(v) \setminus \{q\}| \ge 3$.
Note that if $v$ is adjacent in $G_1$ with a vertex $x\in V(W)$, then $vx \in E(W)$.
Hence we can assume that $N_{G_1}(v) \setminus \{q\}$ contains at least one vertex outside $V(W)$, for otherwise the claim holds. 
Let $a$ denote such a neigbor.

Since $a \not \in I_1 \cup I_{\ge 4} \cup S$, we infer that $\deg_{G_1}(a)=3$. Let $x$ and $y$ be the two other neighbors of $a$ except from $v$.
Since $a \not \in S$, we know that $x,y \in S \subseteq V(W)$. By our construction $vx, vy \in E(W)$, so $\deg_W(v) \ge 2$.

Now assume that $(N_{G_1}(v) \setminus \{q\}) \setminus V(W) = \{a\}$. It follows that $|N_{G_1}(v) \cap V(W)| \ge \deg_{G_1}(v)-2\ge 2$. 
If $N_{G_1}(v) \cap V(W)$ contains a vertex $z\not\in\{x,y\}$, then $vz \in E(W)$ and the claim holds. 
It follows that $N_{G_1}(v) \cap V(W) = \{x,y\}$ and $N_{G_1}(v)=\{a,x,y,q\}$.
Hence the vertices $v, x, y, a, q$ induce the configuration from Rule~\ref{r:two-triangles}, a contradiction.

Finally assume that $|(N_{G_1}(v) \setminus \{q\}) \setminus V(W)| \ge 2$. Consider an arbitary vertex $b \in (N_{G_1}(v) \setminus \{q\}) \setminus V(W)$, $b\ne a$. As shown above, $\deg_{G_1}(b)=3$ and if $x'$ and $y'$ denote the two other neighbors of $b$ except from $v$, then $x',y' \in V(W)$ and $vx', vy' \in E(W)$. It follows that $\{x,y\} = \{x',y'\}$, for otherwise the claim holds. 
It implies that when $|(N_{G_1}(v) \setminus \{q\}) \setminus V(W)| \ge 3$, then $G_2$ contains $K_{3,3}$ as a subgraph, a contradiction with planarity.
Hence, $\{q, a, b\} \subseteq N_{G_1}(v) \subseteq \{q, a, b, x, y\}$.
Note that a removal of $v$ disconnects $q$ from the rest of the graph $G_1$.
Thus, as Rule \ref{r:K_2,3} is not applicable for vertices $v, x, y, a, b$, $G_1 \setminus \{x,y\}$ is connected.
However, $N_{G_1}(\{a, b, v, q\}) = \{x,y\}$ and $y$ is not of degree $2$, as otherwise Rule \ref{r:2->1} would be applicable.
We infer that $G_1$ is isomorphic to the graph from Figure \ref{fig:exception}.
This finishes the proof of the claim.

Now we return to the proof of Lemma~\ref{lem:key}. 
If $G_1$ is the graph from Fig.~\ref{fig:exception} we see that $|S_{\ge 3}| = 2$ and $|S|=3$ so the lemma holds.
Hence by the above claim we can assume that $\deg_W(v)\ge 3$ for any $v\in S_1$.
Since also each vertex $v$ in $S_{\ge 3}$ has at least three neighbors $v_1, v_2, v_3$ in $W$, we conclude that for any $v\in S$ we have $\deg_W(v) \ge 3$. By Lemma~\ref{lem:matching}, graph $W$ contains a matching $M$ of size
\begin{equation}
\label{eq:M-lower}
 |M| \ge |S|/3.
\end{equation}
The edges of $M$ are of two kinds:
\begin{itemize}
 \item edges incident with a vertex in $S_{\ge 3} \cup I_{\ge 4}$ (there are at most $|S_{\ge 3}| + |I_{\ge 4}|$ of such edges),
 \item edges with both endpoints in $S_1$ (there are at most $|M^*|$ of such edges).
\end{itemize}
Hence, 
\begin{equation}
\label{eq:M-upper}
 |M| \le |S_{\ge 3}| + |I_{\ge 4}| + |M^*|.
\end{equation}
By combining~\eqref{eq:M-lower} with~\eqref{eq:M-upper} we get the claim of the lemma.
\qed
\end{proof}

\begin{theorem}
 \label{th:bound}
Let $G_2$ be the graph obtained after the kernelization algorithm and let $S$ be any minimum connected vertex cover of $G_2$. 
Then $|V(G_2)| \le \frac{11}{3} |S|$.
\end{theorem}

\begin{proof}
Consider the bipartite planar graph $B$ which consists of the edges of $G_1$ between $S$ and $I_3\cup I_{\ge 4}$. 
Then $3|I_3| + 4|I_{\ge 4}| \le |E(B)| < 2 (|S|+|I_3|+|I_{\ge 4}|)$. This implies that 
\begin{equation}
\label{eq:i3+i4}
|I_3|+|I_{\ge 4}| < 2|S| - |I_{\ge 4}|. 
\end{equation}
By Lemma~\ref{lem:key}, $|S_1| \le 2 |S_{\ge 3}| + 3 |I_{\ge 4}| + 3|M^*|$. Using this we get
\begin{equation}
\label{eq:i1}
|I_1|=|S_1|=\frac{2}{3}|S_1|+\frac{1}{3}|S_1|\le \frac{2}{3}|S_1|+\frac{2}{3}|S_{\ge 3}| + |I_{\ge 4}| + |M^*|. 
\end{equation}
Now we are ready to bound the number of vertices in $G_2$:
\begin{eqnarray*}
 |V(G_2)| & = & |S| + |I_1| - |M^*| + |I_3| + |I_{\ge 4}| \\
          & \le^{\text{\eqref{eq:i3+i4},\eqref{eq:i1}}} & |S| + \frac{2}{3}|S_1|+\frac{2}{3}|S_{\ge 3}|+2|S| \\
          & = & \frac{11}{3}|S|.
\end{eqnarray*}
 \qed
\end{proof}

\section{An example with tight analysis}
\label{sec:example}

In this section we show an example of a planar graph where the analysis from the previous section is tight.
That is, we construct a graph $G$ with the following properties:
no reduction of Phase 1 is applicable, $G$ admits a connected vertex cover
of size roughly $|V(G)|/4$, and Rule \ref{r:1->2} may be used at most $|V(G)|/12$ times.

Consider a gadget graph $H$ depicted on Figure \ref{fig:example}. For any integer $\ell \geq 3$, the graph $G_\ell$ is constructed by taking $\ell$ copies of $H$
and connecting them in the following manner:
\begin{enumerate}
\item In all copies of $H$, all vertices $s$ are identified into a single vertex; similarly, all vertices $t$ are identified into a single vertex.
\item The vertex $v$ from the $i$-th copy of $H$ is identified with the vertex $u$ from the $(i+1)$-th copy of $H$ and the vertex $v$ from the last copy of $H$
is identified with the vertex $u$ from the first copy of $H$; moreover, the $1$-neighbours of the aforementioned pairs vertices are also identified.
\item Any multiple edges, resulting in the above operations, are removed.
\end{enumerate}

\begin{figure}
 \begin{center}
\includegraphics{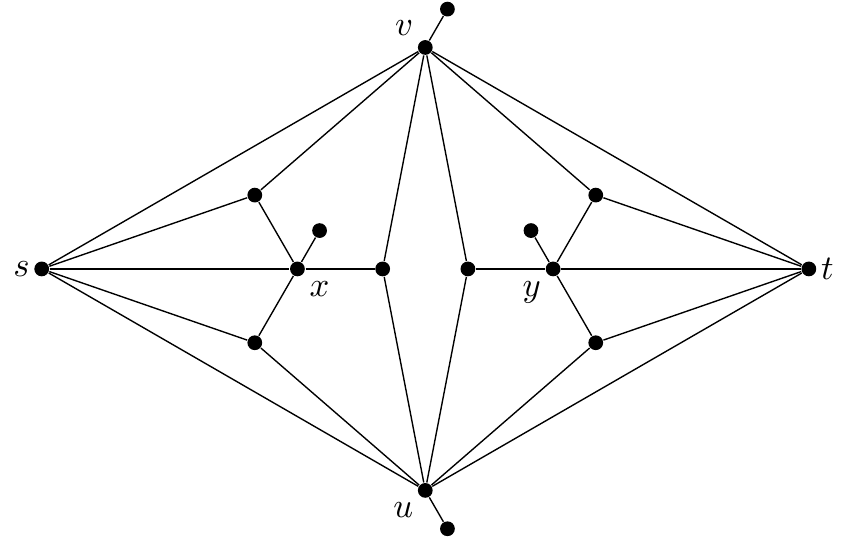}
\caption{A gadget graph $H$ used in the construction of an example with tight analysis.}
\label{fig:example}
\end{center}
\end{figure}

It is easy to see that the graph $G_\ell$ is planar and a direct check ensures us that Rules 1-$\ref{r:last}$ are not applicable to $G_\ell$.
Now note that the set $S$ consisting of vertices $s$, $t$ and all vertices $x$, $y$, $u$ and $v$ is a minimum connected vertex cover
of $G_\ell$. Moreover, $|V(G_\ell)| = 12\ell+2$ and $|S| = 3\ell+2$. Let us analyze sets $S_1$,
$S_{\geq 3}$, $I_1$, $I_3$ and $I_{\geq 4}$ (defined as in the previous section) in the graph $G$.
\begin{enumerate}
\item $S_{\geq 3} = \{s,t\}$, thus $|I_1| = |S_1| = 3\ell$, i.e., almost every vertex in $S$ has a $1$-neighbour;
\item $I_{\geq 4} = \emptyset$ and $|I_3| = 6\ell$, i.e., $|I_3| + |I_{\geq 4}| = 2|S| - 4$;
\end{enumerate}

Finally, let us analyze how many times Rule \ref{r:1->2} can be applied to $G_\ell$. Note that no pair of vertices $x$ and $y$ lie on the same face of the graph $G_\ell$,
thus any edge in the matching $M^*$ (constructed in Phase 2) needs to have an endpoint in a vertex $u$ or $v$. There are $\ell$ such vertices, thus $|M^*| \leq \ell$
(in fact it is easy to see that $|M^*| = \ell$, as we can match $u$ to $x$ in every copy of $H$). We conclude that
$$|S_{\geq 3}| + |I_{\geq 4}| + |M^*| = 2 + 0 + \ell = \frac{|S|+4}{3},$$
and the bound in Lemma \ref{lem:key} is tight up to an additive constant.

\bibliographystyle{splncs03}
\bibliography{planar-cvc-kernel}

\end{document}